\documentclass[11pt]{llncs}
\usepackage{microtype}
\usepackage{makeidx}
\usepackage{graphicx}
\usepackage{xcolor}
\usepackage{ifthen}
\usepackage{microtype}
\usepackage{boxedminipage}
\usepackage{amsmath}
\usepackage{amsfonts}
\usepackage{amssymb}
\usepackage{authblk}
\usepackage{breakcites}
\usepackage{chemscheme}
\usepackage{framed}
\usepackage{float}
\usepackage{breqn}
\usepackage{array, makecell}

\newcommand{\tinym}[1]{{\tiny{\mbox{#1}}}}

\newcommand{\E}{{\bf E}}

\newcommand{\etal}{{\it et al. }}

\newcommand{\veps}{\varepsilon}

\newcommand{\FN}[1]{||#1||_F^2}
\newcommand{\AV}[1]{\mathbf{A}^{(#1)}}
\newcommand{\XV}[1]{\mathbf{X}^{(#1)}}
\newcommand{\rv}[1]{\mathbf{r}^{(#1)}}
\newcommand{\sv}[1]{\mathbf{s}^{(#1)}}
\newcommand{\tv}[1]{\mathbf{t}^{(#1)}}
\newcommand{\zv}[1]{\mathbf{z}^{(#1)}}
\newcommand{\norm}[1]{\lVert #1 \rVert}
\newcommand{\half}{\frac{1}{2}}

\newcommand{\TA}{\tilde{A}}
\DeclareMathOperator*{\argmin}{argmin}

\newtheorem{fact}{Fact}

\allowdisplaybreaks

\graphicspath{{./Figures/}}

\begin{document}

\title{Multiplicative Rank-$1$ Approximation using Length-Squared Sampling}
%
%
\author{Ragesh Jaiswal\thanks{Part of the work was done while the author was on sabbatical from IIT Delhi and visiting UC San Diego.} 
\and Amit Kumar}
%
%
%
\institute{
Department of Computer Science and Engineering, \\
Indian Institute of Technology Delhi.\thanks{Email addresses: \email{\{rjaiswal, amitk\}@cse.iitd.ac.in}}
}
{\def\addcontentsline#1#2#3{}\maketitle}

\begin{abstract}
We show that the span of $\Omega(\frac{1}{\veps^4})$ rows of any matrix $A \subset \mathbb{R}^{n \times d}$ sampled according to the length-squared distribution contains a rank-$1$ matrix $\tilde{A}$ such that $\FN{A - \tilde{A}} \leq (1 + \veps) \cdot \FN{A - \pi_1(A)}$, where $\pi_1(A)$ denotes the best rank-$1$ approximation of $A$ under the Frobenius norm. Length-squared sampling has previously been used in the context of rank-$k$ approximation. However, the approximation obtained was additive in nature. We obtain a multiplicative approximation albeit only for rank-$1$ approximation.
\end{abstract}




\pagestyle{plain}
\setcounter{page}{1}

\section{Introduction}
Rank-$k$ approximation is an important problem in data analysis.
Given a dataset represented as a matrix $A \subset \mathbb{R}^{n \times d}$, where the rows of the matrix represent the data points, the goal is to find a rank-$k$ matrix $\tilde{A}$  such that $\FN{A - \tilde{A}}$ is is not too large compared to $\FN{A - \pi_k(A)}$. 
Here $\pi_k(A)$ denotes the best rank-$k$ matrix under the Frobenius norm, that is, $$\pi_k(A) = \argmin_{{X: \mbox{\small{rank}}(X) \leq k}} \FN{A - X}.$$
Note that this problem is not computationally hard and can be solved using Singular Value Decomposition (SVD). Here, we discuss a simpler sampling based algorithm.
From a geometric perspective, the problem is to find a {\em best-fit} $k$-subspace to a given $n$ points in $d$-dimensional Euclidean space where the measure of fit is the sum of squared distance of the points to the subspace.
In this article we restrict our discussion to rank-$1$ approximation which corresponds to the geometric {\em best-fit line} problem. 
We will be using the matrix and the geometric interpretations interchangeably in this discussion.

We will discuss a sampling technique for this problem. 
We start with the following question: ``{\it Is there a simple sampling procedure that samples a few rows of the given matrix such that the span of the sampled rows contain a good rank-1 approximation?}"
Let us try the simplest option of sampling from the uniform distribution.
One quickly realises that it is easy to construct datasets where the span of even a fairly large sample of uniformly sampled rows does not contain a good rank-$1$ matrix.
For example, consider a two-dimensional dataset where all points  except one  have coordinate $(0, y)$ and the remaining point has coordinate $(x, 0)$ and $x \gg y$.
This example suggests that one should take the norm of a point into consideration while sampling.
This naturally leads to {\em length-squared sampling}. 
The idea is to sample rows such that the sampling probability of the $i^{th}$ row is proportional to the square of its norm.
That is, the sampling probability $p_i$ of the row $\AV{i}$ of a matrix $A$ is given by:
\[
p_i = \frac{\norm{\AV{i}}^2}{\FN{A}}
\]
Length-squared sampling has been explored in the past work of Frieze \etal~\cite{fkv04} and further explored in various works~\cite{drvw06,dv06}.
The main result known from previous works in the context of rank-$1$ approximation is the following theorem of Frieze \etal~\cite{fkv04}.

\begin{theorem}[\cite{fkv04}]\label{thm:frieze}
Let $0 < \veps < 1$. Let $S$ be a sample of $s$ rows of an $n \times d$ matrix $A$, each chosen independently with length-squared distribution. If $s = \Omega(\frac{1}{\veps})$, then the span of $S$ contains a matrix $\TA$ of rank-$1$ such that:
\[
\E[\FN{A - \TA}] \leq \FN{A - \pi_1(A)} \ +\  \veps \cdot \FN{A}.
\]
\end{theorem}
Note that this only gives an additive approximation and the additive error of $\veps \cdot \FN{A}$ can be very large since we do not have any control on $\FN{A}$.
This raises the question about whether a multiplicative approximation could be possible.
The subsequent works of Deshpande and Vempala~\cite{dv06} and Deshpande \etal~\cite{drvw06} use {\em adaptive length-squared sampling} along with {\em volume sampling} to obtain a multiplicative approximation.
In this work, we show that length-squared sampling is sufficient to obtain a multiplicative approximation albeit at the cost of using a slightly larger sample size.
Our main result is formally stated as the following theorem.

\begin{theorem}[Main result]\label{thm:main}
Let $0 < \veps < 1$. Let $S$ be a sample of $s$ rows of an $n \times d$ matrix $A$, each chosen independently with length-squared distribution. If $s = \Omega(\frac{1}{\veps^4})$, then the span of $S$ contains a matrix $\TA$ of rank-$1$ such that:
\[
\E[\FN{A - \TA}] \leq (1+\veps) \cdot \FN{A - \pi_1(A)}.
\]
\end{theorem}
We prove our main result in the next section. 
Before we do this, let us discuss the application of our results in the {\em streaming setting} that is relevant for big data analysis where $n$ and $d$ are very large\footnote{In this setting, one is allowed to make a few passes over the dataset while using limited amount of workspace. That is, the amount of space the should not scale linearly with the dataset size. This makes sense for big data analysis where it may not be possible to hold the entire dataset in the memory for processing.
}.
Note that length-squared sampling will naturally give a 2-pass streaming algorithm that uses $O(\frac{n+d}{\veps^4}\log{nd})$ space. 
Here, in the first pass, we perform length-squared sampling 
using {\em reservoir sampling}\footnote{In order to maintain a single sample one does the following. The first row is stored with probability $1$. On seeing the $i^{th}$ row ($i > 1$), the stored item is replaced wth $\AV{i}$ with probability $\frac{\norm{\AV{i}}^2}{\sum_{j = 1}^{i} \norm{\AV{j}}^2}$. A simple telescoping product shows that the rows get sampled with the desired probability.}.
In the second pass, we project all the points in the space spanned by the sampled rows and find the best fit line in this smaller dimensional space. 
It is important to note that a streaming algorithm with similar space bound that works using only one pass is known~\cite{clarkson}. 
So, our note is more about the properties of length-squared sampling than streaming algorithms for rank-1 approximation.

\subsection{Related work}
Low rank approximation of matrices has large number of applications in information retrieval and data mining (see e.g.~\cite{dkr02,prtv00,afkm01,dfkvv04}). There has been lot of recent activity in obtaining low rank approximations in time depending on the number of non-zero entries in the input matrix~\cite{cw17,s06,dv06,p14}. All of these methods rely on computing suitable random projections of the input matrix. 
Length-squared sampling is a natural sampling algorithm and has had applications in many problems involving matrix computations~\cite{dk01,dk03,fkv04,dfkvv04}. 
As mentioned earlier, Frieze \etal~\cite{fkv04} showed that this can also be used for obtaining low rank approximations, but one needs to incur an additive error term. 
This restriction was removed in subsequent works~\cite{dv06,s06,dmm06,dmm06b,dmm06c,ndt09,mz11,dmms11} but using different techniques. 
Our main contribution is to show that length-squared sampling is sufficient to obtain a bounded multiplicative error for rank-1 approximation.

\section{Rank-$1$ approximation}
We prove our main result in this section.
Before delving into the proof, we give some intuition behind the analysis. 
By a suitable rotation, we can assume that $\pi_1(A)$, the closest rank-$1$ matrix to $A$ in terms of Frobenius norm, is the first column of $A$. 
Let $\sigma^2$ and $r^2$ denote $\FN{\pi_1(A)}$ and
$\FN{A-\pi_1(A)}$ respectively. If $r$ is large compared to $\sigma$, then the additive guarantee given by Theorem~\ref{thm:frieze} implies a multiplicative guarantee as well. So the non-trivial case is when $r \ll \sigma.$ Let $r_i$ and $\sigma_i$ denote the contribution towards $r$ and $\sigma$ from the $i^{th}$ row respectively. So for most rows, $r_i \ll \sigma_i$ -- we call such rows {\em good} rows. When we sample a good row, the normalized row will be close to the vector $(1, 0, \ldots, 0)$. The heart of the analysis relies on showing that  the {\em average} of such normalized sampled rows will be close to $(1, 0, \ldots, 0)$. In other words, we need to bound the variance term corresponding to rows which are not good. 

Let $\AV{i}$ denote the $i^{th}$ row of matrix $A$. 
Let $\mathbf{v}$ denote the unit vector such that $\norm{A\mathbf{v}}^2$ is maximised.
Note that $\mathbf{v}$ is the largest singular vector of matrix $A$.
We assume without loss of generality that $\mathbf{v} = (1, 0, 0, ..., 0)$\footnote{For a matrix $A$ and a unitary matrix $Q$ of appropriate dimension, the lengths of the corresponding rows in $A$ and $AQ$ are the same. So we can choose a suitable $Q$ such that $\mathbf{v}$ has the mentioned property}.
Let $\sigma^2 = \norm{A\mathbf{v}}^2$.
So, we can write $\AV{i} \equiv (\sigma u_i, \rv{i})$, where $\sum_i u_i^2 = 1$ and $\rv{i}$ is a vector of dimension $(d-1)$.
Let $r_i \equiv \norm{\rv{i}}$ and $r^2 \equiv  \sum_i \norm{\rv{i}}^2 = \sum_i r_i^2$.
The following lemma states that under the assumption that $r^2$ is significantly larger than $\sigma^2$, the conclusion of our main theorem holds due to Theorem~\ref{thm:frieze} (as stated  in~\cite{fkv04}).

\begin{lemma}\label{lemma:fkv-lemma}
If $r^2 > \veps^3 \sigma^2$, then there is a rank-1 matrix $\TA$ in the span of $\Omega \left(\frac{1}{\veps^4}\right)$ independently sampled rows of $A$ sampled with length-squared distribution such that  $\E[\FN{A - \TA}] \leq (1 + \veps) \cdot \FN{A - \pi_1(A)}$.
\end{lemma}
\begin{proof}
Note that since $\mathbf{v}$ maximises $\norm{A\mathbf{v}}^2$, we have $\pi_1(A) = \left(\begin{smallmatrix} \sigma u_1, 0, ..., 0\\\vdots \\ \sigma u_2, 0, ..., 0 \end{smallmatrix}\right),$ which implies that $\FN{A- \pi_1(A)} = r^2$. Also, $\FN{A} = \sigma^2 + r^2$. Combining the above with Theorem~\ref{thm:frieze} (where we use $\frac{\veps^4}{2}$ for $\veps$), we get:
\begin{eqnarray*}
\E[\FN{A - \TA}] \leq r^2 + \frac{\veps^3}{2} \cdot (r^2 + \sigma^2) 
\leq (1 + \veps) \cdot r^2 
= (1 + \veps) \cdot \FN{A - \pi_1(A)}.
\end{eqnarray*}
This completes the proof of the lemma.\qed
\end{proof}
For the remainder of the proof, we will assume that 
\begin{equation}\label{eqn:assume}
r^2 \leq \veps^3 \sigma^2.
\end{equation}
Let $\mathbf{s}$ be a randomly sampled row of matrix $A$ sampled with length-squared distribution and let $\sv{1}, \sv{2}, ..., \sv{l}$ be $l$ independent copies of $\mathbf{s}$.
We would like to define a deterministic procedure to construct a (random) rank-1 matrix $X$ using $\sv{1}, ..., \sv{l}$, where each row of $X$ lies in the span of $\sv{1}, ..., \sv{l}$, such that the expected value of $\FN{A - X}$ is at most $(1+\veps) \cdot r^2$.
Another (geometric) way of saying this is that there is a point $\mathbf{t}$ in the span of $\sv{1}, ..., \sv{l}$ such that the squared distance of rows of $A$ from the line $\ell(\mathbf{t})$ is at most $(1+\veps)$ times of that from the best-fit line $\ell(\mathbf{v})$. Here $\ell(.)$ denotes the line passing through the given point and origin $o = (0, ..., 0)$.

We will need a few more definitions to give the procedure that defines $X$ from $\sv{1}, ..., \sv{l}$.
We first divide the rows into ``good" and ``bad". A row $\AV{i}$ is said to be good if 
\begin{equation}\label{eqn:good-row}
r_i^2 < \veps \sigma^2 u_i^2,
\end{equation}
otherwise it is bad. 
We now give the procedure for mapping the randomly length-squared sampled $\sv{1}, ..., \sv{l}$ to an appropriate matrix $X$.

\begin{framed}
\hspace*{-0.4in}\ \ \ {\tt Mapping($\sv{1}, ..., \sv{l}$)}\\
\hspace*{0.0in}  \ \ \ - For all $i \in \{1, ..., l\}$: \\
\hspace*{0.3in}  \ \ \ - If ($\sv{i}$ is a bad row) then $\tv{i} \leftarrow (0, ..., 0)$\\
\hspace*{0.3in}  \ \ \ - Else $\tv{i} \leftarrow \frac{\sv{i}}{\sigma u_i}$\\
\hspace*{0.0in} \ \ \ - $\mathbf{t} \leftarrow \frac{\sum_{i=1}^{l} \tv{i}}{l}$\\
\hspace*{0.0in} \ \ \ - For all $i \in \{1, ..., n\}$: $\XV{i} \leftarrow \sigma u_i \mathbf{t}$\\
\hspace*{0.0in} \ \ \ - $X \leftarrow \left(\begin{smallmatrix} \XV{1} \\ \vdots \\ \XV{n}\end{smallmatrix} \right)$
\end{framed}

The intuition is that if $\sv{i}$ is a good row, then $\tv{i}$ will be close to 
$\mathbf{v} = (1,0, \ldots, 0)$. So $X$ will be very close to $\pi_1(A)$. 
Note that the above defined procedure is only meant for the analysis and is never actually executed.
Also it is easy to see that the $n \times d$ matrix $X$ defined above is a rank-1 matrix.
We will now bound $\E[\FN{A - X}]$, which is the same as $\sum_i \E[\FN{\AV{i} - \XV{i}}]$.
We start with a simple lemma that bounds the probability of sampling a bad row.

\begin{lemma}\label{lem:bad}
The probability that a sampled row $\mathbf{s}$, sampled using length-squared distribution, is bad is at most $\frac{2r^2}{\veps \sigma^2}$.
\end{lemma}
\begin{proof}
The probability of sampling a bad row $\AV{i}$ is given by $\frac{r_i^2 + \sigma^2 u_i^2}{\sigma^2 + r^2} \leq \frac{r_i^2 + \sigma^2 u_i^2}{\sigma^2}$. 
So, the probability that a sample row is bad is at most $\frac{\sum_{i \textrm{\ is bad}} (r_i^2 + \sigma^2 u_i^2)}{\sigma^2} \leq \frac{\sum_{i \textrm{\ is bad}} (1 + \frac{1}{\veps}) r_i^2 }{\sigma^2} \leq \frac{2 r^2}{\veps \sigma^2}$.\qed
\end{proof}

Let $\AV{i}_1$ denote the first coordinate of $\AV{i}$. Define $\XV{i}_1$ and $\tv{i}_1$ similarly.
We first estimate $\sum_i \E[(\AV{i}_1 - \XV{i}_1)^2]$.

\begin{lemma}\label{lemma:first-coordinate}
$\sum_i \E[(\AV{i}_1 - \XV{i}_1)^2] \leq 5 \veps r^2$.
\end{lemma}
\begin{proof}
Fix an index $i$. Note that $\AV{i}_1$ is $\sigma u_i$ and $\XV{i}_1$ is $\sigma u_i \mathbf{t}_1$. Therefore:
\begin{eqnarray*}
(\AV{i}_1 - \XV{i}_1)^2 = \sigma^2 u_i^2 (1 - \mathbf{t}_1)^2 
\leq \frac{\sigma^2 u_i^2}{l^2} \left( \sum_{j=1}^{l} (1 - \tv{j}_1)\right)^2.
\end{eqnarray*}
The last inequality follows from the Jensen's inequality using $(1-x)^2$ as the convex function.
We obtain a bound on the expectation of $(\AV{i}_1 - \XV{i}_1)^2$ from the above.
\begin{equation}\label{eqn:3}
\E[(\AV{i}_1 - \XV{i}_1)^2] \leq \frac{\sigma^2 u_i^2}{l^2} \cdot \left( \sum_{j=1}^{l} \E \left[\left(1 - \tv{j}_1 \right)^2 \right] \right) +  \frac{\sigma^2 u_i^2}{l^2} \cdot \left( \sum_{j=1}^{l}\E \left[\left(1 - \tv{j}_1 \right) \right]\right)^2
\end{equation}
The previous inequality follows from the independence of random variables $\tv{1}_1, ..., \tv{l}_1$.
Lemma~\ref{lem:bad} shows that the probability of the Bernoulli random variable $\tv{j}_1$ being $0$ for any $j$ is at most $\frac{2r^2}{\veps \sigma^2}$. 
So, $\E[(1 - \tv{j}_1)] \leq \frac{2r^2}{\veps \sigma^2}$ and $\E[(1 - \tv{j}_1)^2] \leq \frac{2r^2}{\veps \sigma^2}$ (note that  $\tv{j}_1$ is either 0 or 1). Substituting this in~(\ref{eqn:3}) above, we get 
\[
\E[(\AV{i}_1 - \XV{i}_1)^2] \leq \frac{2 u_i^2 r^2}{\veps l} + \frac{4 u_i^2 r^4}{\veps^2 \sigma^2}
\]
The lemma follows from the facts that $\sum_i u_i^2 = 1$, $r^2 \leq \veps^3 \sigma^2$ and $l \geq \frac{2}{\veps^4}$.\qed
\end{proof}

Now we estimate the contribution towards $\sum_i \E[\norm{\AV{i} - \XV{i}}^2]$ from coordinates other than the first coordinate.
If $\mathbf{z}$ denotes the vector obtained from $\mathbf{t}$ by removing the first coordinate, then observe that $\sum_i \sum_{d=2}^{l} (\AV{i}_j - \XV{i}_j)^2 = \sum_i \norm{\rv{i} - \sigma u_i \mathbf{z}}^2$.
Now, observe that
\begin{equation}\label{eqn:4}
\sum_i \E[\norm{\rv{i} - \sigma u_i \mathbf{z}}^2] = r^2 - 2 \sigma \langle \sum_i u_i \rv{i}, \E[\mathbf{z}] \rangle + \sigma^2 \E[\norm{\mathbf{z}}^2]
\end{equation}

We will now estimate each of the terms above. 
First, note that if $\zv{j}$ denotes the vector obtained from $\tv{j}$ by removing the first coordinate, then $\mathbf{z} = \frac{1}{l} \sum_{j=1}^{l} \zv{j}$.
Let $G$ denote the index set of good rows. Then observe that for any $j$,
\begin{equation}\label{eqn:5}
\E[\zv{j}] = \sum_{k \in G} p_k \frac{\rv{k}}{u_k \sigma},
\end{equation}
where $p_k$ is the probability of length-squared sampling the $k^{th}$ row of $A$.
Since the vectors $\tv{j}$ are chosen i.i.d, the expectation of $\mathbf{z}$ can also be written in terms of the above expression (i.e., RHS of (\ref{eqn:5})).
Next, we show a useful inequality for good rows.

\begin{fact}\label{lemma:good-p}
If row $\AV{i}$ if good, then $|u_i - \frac{p_i}{u_i}| \leq \veps |u_i|$.
\end{fact}
\begin{proof}
Assume that $u_i > 0$, otherwise we can replace $u_i$ by $-u_i$ in the following argument. Since row $\AV{i}$ is good, we know that $r_i^2 \leq \veps \sigma^2 u_i^2$. Also note that length squared sampling means that $p_i = \frac{\sigma^2 u_i^2 + r_i^2}{\sigma^2 + r^2}$. Using these we get:
\[
\frac{p_i}{u_i} \leq \frac{\sigma^2 u_i^2 + r_i^2}{\sigma^2 u_i} = u_i + \frac{r_i^2}{\sigma^2 u_i} \leq (1+\veps) \cdot u_i,
\]
and
\[
\frac{p_i}{u_i} \geq \frac{\sigma^2 u_i}{r^2 + \sigma^2} \stackrel{(\ref{eqn:assume})}{\geq} \frac{u_i}{\veps^3 + 1} \geq (1-\veps) \cdot u_i
\]
The lemma follows from the above two inequalities.\qed
\end{proof}

\noindent

The next lemma gives an upper bound on $\norm{\sigma \cdot \E[\mathbf{z}] - \sum_i u_i \rv{i}}$.

\begin{lemma}\label{lemma:mid-1}
$\norm{\sigma \cdot \E[\mathbf{z}] - \sum_i u_i \rv{i}} \leq 2 \veps r$.
\end{lemma}
\begin{proof}
Using the triangle inequality and an expression for $\E[\mathbf{z}]$ using (\ref{eqn:5}), we get that  : 
\begin{eqnarray*}
\norm{\sigma \cdot \E[\mathbf{z}] - \sum_i u_i \rv{i}}
&\leq& \norm{\sum_{i \notin G} u_i \rv{i}} + \norm{\sum_{i \in G} \left(u_i - \frac{p_i}{u_i} \right) \rv{i}}\\
& \leq &  \sum_{i \notin G} |u_i| r_i + 
\sum_{i \in G} |u_i - \frac{u_i}{p_i}| r_i \\
& \stackrel{\mbox{\tiny{Fact~\ref{lemma:good-p}}}}{\leq} &  \sum_{i \notin G} |u_i| r_i + 
\sum_{i \in G} \veps |u_i| r_i \\
&\stackrel{\tinym{Cauchy-Schwarz}}{\leq}& \left( \sum_{i \notin G} u_i^2\right)^{\frac{1}{2}} \left( \sum_{i \notin G} r_i^2\right)^{\frac{1}{2}} + \\
&& \qquad \veps \cdot \left( \sum_{i \in G}  u_i^2 \right)^{\half} \left( \sum_{i \in G} r_i^2\right)^{\half}\\
&\stackrel{ (\ref{eqn:good-row})}{\leq}& r \cdot \left( \sum_{i \notin G} \frac{r_i^2}{\veps \sigma^2}\right)^{\half} + \veps r \\
&\leq& r \cdot \left( \frac{r^2}{\veps \sigma^2}\right)^{\half} + \veps r \\
&\stackrel{ (\ref{eqn:assume})}{\leq}&2 \veps r 
\end{eqnarray*}
This completes the proof of the lemma.\qed
\end{proof}

\noindent
We now show a useful fact regarding $\rv{i}$.

\begin{fact}\label{lemma:fact}
$\norm{\sum_i u_i \rv{i}} \leq r$.
\end{fact}
\begin{proof}
The statement follows by triangle inequality and Cauchy-Schwarz:
\[
\norm{\sum_i u_i \rv{i}} \leq \sum_i |u_i|r_i \leq 
\left( \sum_i u_i^2 \right)^{\frac{1}{2}} \cdot \left( \sum_i r_i^2\right)^{\frac{1}{2}} = r. 
\]
This completes the proof.\qed
\end{proof}

\noindent
The next lemma bounds the middle term of the RHS of~(\ref{eqn:4}).

\begin{lemma}\label{lemma:term-2}
$2 \sigma \cdot \langle \sum_i u_i \rv{i}, \E[\mathbf{z}]\rangle \geq 2 \norm{\sum_i u_i \rv{i}}^2 - 4 \veps \cdot r^2$.
\end{lemma}
\begin{proof}
We have
\begin{eqnarray*}
\sigma \cdot \langle \sum_i u_i \rv{i}, \E[\mathbf{z}]\rangle = \langle \sum_i u_i \rv{i}, \sigma \cdot \E[\mathbf{z}]\rangle 
= \langle \sum_i u_i \rv{i}, \sum_j u_j \rv{j}\rangle + \langle \sum_i u_i \rv{i}, \sigma \cdot \E[\mathbf{z}] - \sum_j u_j \rv{j}\rangle
\end{eqnarray*}
We now get the following inequalities:
\begin{eqnarray*}
\left\vert \langle \sum_i u_i \rv{i}, \sigma \cdot \E[\mathbf{z}] - \sum_j u_j \rv{j}\rangle \right\vert 
\stackrel{\tinym{Cauchy-Schwarz}}{\leq} 
\norm{\sum_i u_i \rv{i}} \cdot \norm{\sigma \cdot \E[\mathbf{z}] - \sum_j u_j \rv{j}} 
\stackrel{\tinym{Lemma~\ref{lemma:mid-1}, Fact~\ref{lemma:fact}}}{\leq} 2 \veps r^2 
\end{eqnarray*}
The lemma follows from the above inequality.\qed
\end{proof}

We now bound the last RHS term of~(\ref{eqn:4}).
Since $\mathbf{z} = \frac{\zv{1} + \zv{2} + ... + \zv{l}}{l}$, it is easy to see that
\begin{equation}
\sigma^2 \E[\norm{\mathbf{z}}^2] \leq \frac{\sigma^2}{l} \E[\norm{\zv{j}}^2] + \sigma^2 \cdot \left( \norm{\E[\zv{j}]}\right)^2,
\end{equation}
where $j$ is an arbitrary index between $1$ and $l$.
We can bound the two terms on the RHS below. 

\begin{lemma}\label{lemma:term-3-1}
For any index $j$, $\frac{\sigma^2}{l} \cdot \E[\norm{\zv{j}}^2] \leq \veps r^2$.
\end{lemma}
\begin{proof}
Note that $\zv{j}$ is equal to 
$\frac{\rv{k}}{u_k \sigma}$ with probability $p_k$. Therefore, 
We have
\[
\frac{\sigma^2}{l} \cdot \E[\norm{\zv{j}}^2] = \frac{\sigma^2}{l} \cdot \sum_{k \in G} \frac{p_k r_k^2}{u_k^2 \sigma^2} \stackrel{\tinym{Fact~\ref{lemma:good-p}}}{\leq} \frac{(1+\veps) r^2}{l}
\]
The lemma now follows from the fact that $l \geq \frac{2}{\veps^4}$.\qed
\end{proof}

\begin{lemma}\label{lemma:term-3-2}
For any index $j$, $\sigma^2 \cdot (\norm{\E[\zv{j}]})^2 \leq 2 \norm{\sum_{k} u_k \rv{k}}^2 + 4\veps^2 r^2$.
\end{lemma}
\begin{proof}
For any index $k \in G$, let $\delta_k$ denote $\frac{p_k}{u_k} - u_k$. Fact~\ref{lemma:good-p} shows that $|\delta_k| \leq \veps |u_k|$.
For any index $k \notin G$, let $\delta_k$ denote $-u_k$. Using~(\ref{eqn:5}), we can write:
\begin{eqnarray*}
\sigma^2 \cdot (\norm{\E[\zv{j}]})^2 &=& \norm{\sum_{k \in G} p_k \frac{\rv{k}}{u_k}}^2 \\
&=& \norm{ \left(\sum_{k \in G} u_k \rv{k} + \sum_{k \in G} \delta_k \rv{k} \right)}^2   \qquad \textrm{(since $\delta_k = \frac{p_k}{u_k}-u_k$ for $k \in G$)}\\
&=& \norm{ \left(\sum_{k} u_k \rv{k} + \sum_{k} \delta_k \rv{k} \right)}^2  \qquad \textrm{(since $\delta_k = -u_k$ for $k \notin G$)}\\
&\leq& 2\ \norm{\sum_k u_k \rv{k}}^2 + 2\  \norm{\sum_k \delta_k \rv{k}}^2.
\end{eqnarray*}

The second term above can be bounded as follows:
\begin{eqnarray*}
\norm{\sum_k \delta_k \rv{k}}^2 
&\leq& \sum_k \delta_k^2 r_k^2 \\
&\stackrel{\tinym{Fact~\ref{lemma:good-p}}}{\leq}& \veps^2 r^2 \sum_{k \in G} u_k^2  + r^2 \sum_{k \notin G} u_k^2 \\
& \stackrel{(\ref{eqn:good-row})}{\leq} & \veps^2 r^2 + \frac{r^4}{\veps \sigma^2} \\&\stackrel{(\ref{eqn:assume})}{\leq}&  2 \veps^2 r^2.
\end{eqnarray*}

This completes the proof of the lemma.\qed
\end{proof}

Now, combining Lemma~\ref{lemma:term-2}, Lemma~\ref{lemma:term-3-1}, Lemma~\ref{lemma:term-3-2}, we see that~(\ref{eqn:4}) can be simplified as:
\[
\sum_{i} \E[\norm{\rv{i} - \sigma u_i \mathbf{z}}^2] \leq r^2 + 9 \veps r^2.
\]

Combining this with Lemma~\ref{lemma:first-coordinate}, we get that $\FN{A - X}$ has an expected value at most $(1+15\veps) r^2$ for the case where~(\ref{eqn:assume}) holds. 
Finally, combining this with Lemma~\ref{lemma:fkv-lemma}, we obtain the main result in  Theorem~\ref{thm:main}\footnote{The extra factor of $15$ can be handled by using $\veps/15$ instead of $\veps$ in the sampling procedure.}.

\section{Conclusion and Open Problems}
The following questions related to length-squared sampling are relevant for the current discussion.
\begin{enumerate}
\item Does a single length-squared sampled point approximate the best-fit line? How good is the approximation?

\item Does the result also hold for rank-$k$ approximation for $k > 1$? That is, does a set of $f(\veps)$ length-squared sampled rows (for some function $f$ of $\veps$) contain a rank-$k$ matrix $\tilde{A}$ such that $\FN{A - \tilde{A}} \leq (1 + \veps) \cdot \FN{A - \pi_k(A)}$?

\item Our results show that $\Omega(\frac{1}{\veps^4})$ length-squared sampled rows are sufficient to obtain $(1+\veps)$ multiplicative approximation. Can we show that sampling $\Omega(\frac{1}{\veps^4})$ rows are necessary? Note that for additive approximation (see Theorem~\ref{thm:frieze}), $\Omega(\frac{1}{\veps})$ rows are sufficient.

\item Does {\em adaptive} length-squared sampling\footnote{Adaptive sampling means that we sample a sequence of sets $S_1, ..., S_k$, where $S_1$ contains length-squared sampled rows of the given matrix $A$, $S_2$ contains length-squared sampled rows of the matrix $A - \pi_{S_1}(A)$, and so on. Here, $\pi_S(A)$ denotes the projection of the rows of matrix $A$ onto the linear subspace spanned by elements of $S$. See \cite{dv06} for a detailed discussion.}
give multiplicative rank-$k$ approximation?  Note that the previous work of Deshpande and Vempala~\cite{dv06} only showed additive approximation. Multiplicative approximation was shown to be achieved in subsequent work of Deshpande \etal~\cite{drvw06} when length-squared sampling was combined with {\em volume sampling}. So the relevant question is whether length-squared sampling is sufficient for multiplicative approximation.
\end{enumerate}
The first two questions are simple to resolve. 
A simple analysis, which is also implicit in some of the previous works (e.g., \cite{drvw06}), gives us that a single length-squared sampled point gives a 2-factor approximation in expectation.
A simple example gives a negative answer to the second question.
Consider a 2-dimensional point set where all but one point have coordinates $(x, 0)$ and the remaining point has coordinates $(0, y)$, where $x >> y$. In this case, $\FN{A - \pi_2(A)} = 0$ but any constant sized length-squared sampled set is unlikely to contain the point at $(0, y)$. 
The last two questions remain open.

\addcontentsline{toc}{section}{References}
\bibliographystyle{alpha}
\bibliography{paper}

\end{document}